\documentclass[11pt]{article}
\usepackage[margin=1in]{geometry}
\usepackage{algorithmicx}

\usepackage[utf8]{inputenc}
\usepackage{authblk}
\usepackage{amsmath, amssymb, amsthm, thmtools, amsfonts, bm, bbm, thm-restate}
\usepackage{algorithm}
\usepackage{algorithmicx}
\usepackage[noend]{algpseudocode}
\usepackage{cite}
\usepackage[numbers,sort]{natbib}

\usepackage{asymptote}
\usepackage{graphicx}
\usepackage{todonotes}
\usepackage{multirow}
\usepackage{comment}
\usepackage{dsfont}
\usepackage{bigstrut}
\usepackage{caption}
\usepackage{subcaption}
\usepackage{multirow}
\usepackage{makecell}
\usepackage{booktabs}
\usepackage{xcolor}

\usepackage{commath}

\definecolor{BrickRed}{rgb}{0.8,0.25,0.33}
\usepackage[pagebackref,colorlinks,citecolor=blue,linkcolor=BrickRed]{hyperref}

\usepackage{enumitem}
\usepackage[normalem]{ulem}

\usepackage[framemethod=tikz]{mdframed}
\usepackage{nicefrac}

\usepackage{tikz}
\usetikzlibrary{positioning, fit, calc}

\usepackage{pifont}% http://ctan.org/pkg/pifont

\usepackage{cleveref}

\newtheorem{thm}{Theorem}[section]
\newtheorem{fact}[thm]{Fact}
\newtheorem{lem}[thm]{Lemma}
\newtheorem{lemma}[thm]{Lemma}

\newtheorem{rem}[thm]{Remark}

\crefname{thm}{Theorem}{theorems}
\crefname{cla}{Claim}{claims}
\crefname{lem}{Lemma}{lemmas}
\crefname{fact}{Fact}{facts}

\crefname{cor}{Corollary}{corollaries}

\allowdisplaybreaks

\newcommand{\e}{\texttt{e}}

\let\vec\mathbf
\renewcommand{\vec}{\mathbf}

\title{Asymptotically Tight Fractional Online Matching Under Edge Arrivals}
\author{David Wajc\thanks{Supported by  ISF grant 3200/24. Work done while visiting the University of Waterloo.}\\
\vspace{-0.25cm}Technion}
\date{\vspace{-1.5cm}}

\begin{document}

\maketitle

\begin{abstract}
    In this brief note, we close the asymptotic gap between known upper and lower bounds for fractional online matching under edge arrivals. We prove that the optimal competitive ratio for this problem is $1/2+\Theta(1/n)$. The algorithm was suggested and analyzed by OpenAI's ChatGPT Sol based on a single prompt. The presentation was streamlined over a few hours, based on a back and forth discussion with the author, who assumes responsibility for any errors.
\end{abstract}

\section{Introduction}

Online matching is an intensely studied problem (see surveys \cite{mehta2013online,huang2024online}). In the most basic setting, edges of an $n$-node graph arrive online sequentially, and one must decide immediately and irrevocably whether to match each arriving edge. A trivial greedy algorithm achieves a competitive ratio of $1/2$: it outputs a matching of size at least $OPT/2$, 
where $OPT$ denotes the size of a hindsight-optimal matching.
Competitive ratios of $1/2+\Omega(1)$ are achievable if \emph{nodes} are revealed sequentially \cite{karp1990optimal,gamlath2019online,huang2018match}. In contrast, in the stricter edge-arrival model, the optimal competitive ratio is $1/2+o(1)$: \cite{buchbinder2019online,lee2020maximum} achieve a ratio of $1/2+\exp^{-\Theta(n)}$, while \cite{gamlath2019online} show that no ratio above $1/2+O(1/n)$ is attainable, even by \emph{fractional} algorithms. 

Fractional algorithms assign each edge $e$ a value $x_e$ on arrival so that $\vec{x}$ is a \emph{fractional matching}:  $x_e\geq 0$ for each edge $e$ and  $\sum_{f\ni v} x_f\leq 1$ for each~node~$v$.
Such an algorithm has value $ALG\triangleq \sum_e x_e$, and it is \emph{$c$-competitive} if $ALG\geq c\cdot OPT$.

In this note we close this asymptotic gap, giving a $1/2+\Omega(1/n)$-competitive fractional algorithm. 
Together with the result of \cite{gamlath2019online}, this shows that the optimal competitive ratio is $1/2+\Theta(1/n)$.

\section{The Algorithm and Analysis}

We associate each edge $e\in E$ with its time of arrival in the sequence. We therefore use $e<e'$ to mean that $e$ precedes $e'$ in the arrival sequence. We similarly denote by $e+1$ the time immediately after the arrival of $e$ and subsequent assignment of $x_e$. With this notation, we let 
$$\ell_v^e \triangleq \sum_{\substack{e'<e\\e'\ni v}}x_{e'}$$
denote the \emph{load} of node $v$ before time $e$. We also define the \emph{residual capacity} of $v$ as
$$r_v^e \triangleq 1-\ell_v^e.$$
With this notation, the algorithm is particularly simple to state:
\begin{algorithm}[h]
	\caption{Fractional Online Matching}\label{alg}
 \begin{algorithmic}[1]
    \For{\textbf{each} arriving edge $e$, on arrival}
    \State Set: $x_e\gets \min\left\{ r_u^e,\;r_v^e,\;\frac{(r_u^e+r_v^e - \frac{1}{2})^+}{2}\right\},$ where $z^+\triangleq \max\{z,\;0\}.$
    \EndFor
\end{algorithmic}	
\end{algorithm}

Intuitively, after processing an edge, either one endpoint is saturated or its endpoints have combined load at least $3/2$. This dichotomy drives the analysis.

Before proving this intuitive property, we first verify feasibility and monotonicity of the loads.

\begin{fact}[Feasibility]\label{fact:feasible}
    \Cref{alg} outputs a  fractional matching.
\end{fact}
\begin{proof}
   We prove inductively that all residual capacities remain non-negative. Since $x_e\leq r^e_w$ for each $w\in e$, we get that $r^{e+1}_w = r^e_w - x_e \geq 0$. Consequently, $\ell_v^{e+1} \leq 1$ for each vertex $v$ and edge~$e$.
   On the other hand, for each edge $e$, each of the three quantities in the minimum defining $x_e$
    is nonnegative, and so $x_e\geq 0$. 
\end{proof}
Let $\ell_v$ and $r_v\triangleq 1-\ell_v$ denote the load and residual capacity of node $v$ after all edges have their value set. 
    We say a node $v$ is \emph{saturated} if $\ell_v = 1$. Note that $\ell_v = \sum_{e\ni v} x_e \geq \ell_v^e$, with the inequality following by 
non-negativity of $\vec{x}$ (\Cref{fact:feasible}).

The algorithm is designed with the following lemma in mind:
\begin{lemma}[Edge-load dichotomy]\label{fact}
For every arriving edge $e=(u,v)$, the final loads satisfy
\[
\ell_u+\ell_v \geq1.
\]
Moreover, if both $u$ and $v$ are unsaturated at termination, then
\[
\ell_u+\ell_v \geq\frac{3}{2}.
\]
\end{lemma}

\begin{proof}
As $\vec{\ell}$ is non-negative by \Cref{fact:feasible}, the claim is trivial if either endpoint is saturated. Hence, to prove the stronger assertion, suppose both endpoints are unsaturated at termination. By load monotonicity ($\ell_v\geq \ell_v^{e+1}$), neither is saturated immediately after processing $e$. So, $x_e<\min\{r_u^e,\;r_v^e\}$.

If
$0<x_e<\min\{r_u^e,r_v^e\},$
then
$x_e=\frac{r_u^e+r_v^e-\frac12}{2}>0$.
Consequently,
\[
r_u^{e+1}+r_v^{e+1}
=
r_u^e+r_v^e-2x_e
=
\frac12.
\]
Conversely, if $x_e=0$, then since $x_e<\min\{r_u^e,\;r_v^e\}$, we have that $x_e = \frac{(r_u^e+r_v^e - \frac{1}{2})^+}{2} = 0$, and so
\[
r_u^{e+1}+r_v^{e+1} = r_u^e+r_v^e\le \frac12.
\]
In either case, we have by load monotonicity that 
\begin{align*}
\ell_u+\ell_v & \geq\ell_u^{e+1}+\ell_v^{e+1}
=
2-(r_u^{e+1}+r_v^{e+1})
\ge
\frac32. \qedhere
\end{align*}
\end{proof}

The following lemma asserts that if positive mass is assigned to an edge, then at least $1/2$ residual capacity remains across its two endpoints immediately afterward.
This residual slack will later force mass to leave the set of saturated nodes.

\begin{lem}[Residual slack]\label{lem:external-load}
    For each edge $e=(u,v)$ for which $x_e>0$, we have $r_u^{e+1}+r_v^{e+1}\geq \frac{1}{2}.$
\end{lem}
\begin{proof}
    First, the lemma's hypothesis  implies that $\frac{r_u^e+r_v^e-\frac{1}{2}}{2}\geq x_e >0$.
    Now, if $x_e\notin \{r_u^e,\;r_v^e\}$, then the claimed inequality holds with equality, by the choice of $x_e= \frac{r_u^e+r_v^e-\frac{1}{2}}{2}$. Else, let $a=\min\{r_u^e,\;r^v_e\}$ and $b=\max\{r_u^e,\;r^v_e\}$. Then, by our choice~of~$x_e$,
    $$a\leq \frac{a+b-\frac{1}{2}}{2},$$
    from which we obtain that $b-a\geq \frac{1}{2}.$ But then after the update, $x_e\gets a$, we have that 
    \begin{align*}
    r_u^{e+1} + r_v^{e+1} & = r_u^{e} - a + r_v^{e} - a = 0 + b-a \geq \frac{1}{2}. \qedhere
    \end{align*}
\end{proof}

We are now ready to quantify \Cref{alg}'s improvement over the greedy algorithm.

\begin{thm}
    For every graph with $OPT\geq 1$, 
    \Cref{alg} outputs a fractional matching of value 
    $$ALG \geq \frac{OPT}{2}+\frac{1}{4}\geq OPT\left(\frac{1}{2}+\frac{1}{2n}\right).$$
\end{thm}
\begin{proof}
    The second inequality follows from the trivial bound $OPT\leq \frac{n}{2}$. 
    
    It remains to prove the first inequality, $ALG\geq \frac{OPT}{2}+\frac{1}{4}$, or equivalently, as $\sum_v \ell_v = 2\sum_e x_e$, 
    \begin{align}\label{eqn:sufficient}
    \sum_v \ell_v - OPT \geq \frac{1}{2}. 
    \end{align}

    To prove \Cref{eqn:sufficient}, fix a maximum matching $M$ of size $|M|=OPT$. Decomposing the sum across $M$, we have that
    $$\sum_v \ell_v - OPT = \sum_{(u,v)\in M} (\ell_u + \ell_v - 1) + \sum_{w\in V\setminus V(M)} \ell_w.$$
    By \Cref{fact} and \Cref{fact:feasible}, each summand on the right-hand side, $(\ell_u + \ell_v - 1)$ for some $(u,v)\in M$ and $\ell_w$ for some $w\in V\setminus V(M)$, is non-negative. 
    So, if any edge $(u,v)\in M$ has both endpoints unsaturated, then by \Cref{fact}
    , 
    $\ell_u + \ell_v \geq \frac{3}{2}$, and so $\sum_v \ell_v - OPT \geq (\ell_u + \ell_v - 1)\geq \frac{1}{2}.$
    We may therefore assume that each edge in $M$ has at least one  saturated endpoint.
    We will show that in this case at least $1/2$ unit of mass must cross from the saturated nodes to that set's complement, implying the desired bound.

    Denote by $S\triangleq \{v \mid \ell_v = 1\}$ the saturated nodes. Since each of the (node-disjoint) edges in $M$ has at least one endpoint in $S$, we have that $|S|\geq OPT$, and so
    \begin{align}\label{eqn:desired-interesting-case}
    \sum_v \ell_v - OPT \geq \sum_{v\in V\setminus S} \ell_v.
    \end{align}

    Now, if no edge $e\in S\times S$ is assigned $x_e > 0$, then since $OPT\geq 1$, we must assign a value of at least $|S|\geq OPT\geq 1 \geq \frac12$ to edges in $S\times (V\setminus S)$.
   Otherwise, let $e=(u,v)\in S\times S$ be the last edge in $S\times S$ receiving $x_e>0$. By \Cref{lem:external-load}, immediately after assigning $x_e>0$, $$\sum_{v\in S} r^{e+1}_v \geq r^{e+1}_u + r^{e+1}_v \geq \frac{1}{2}.$$
   But, since by definition
    $$\sum_{v\in S}r_v = 0,$$
    we must further decrease the residual capacity of the nodes in $S$ by assigning value of at least $\frac{1}{2}$ to the set of edges in $S\times (V\setminus S)$.
    In either case, using non-negativity of $\vec{x}$ by \Cref{fact:feasible}, $$\sum_{v\in V\setminus S}\ell_v \geq \sum_{e\in S\times (V\setminus S)} x_e \geq \frac{1}{2}.$$
    Combined with \Cref{eqn:desired-interesting-case,eqn:sufficient}, this concludes the proof.
\end{proof}

\begin{rem}
    We leave open the question of closing the constant gap between the precise $\Theta(1/n)$ terms in the competitive ratios of \Cref{alg} ($\frac{1}{2n}$) and the impossibility bound of \cite{gamlath2019online} ($\frac{1}{n+2}$).\footnote{\cite{gamlath2019online} prove a bound of  $\frac{1}{2}+\frac{1}{2n'+2}$ for bipartite graphs with $n'$ nodes on either side, so $n=2n'$ nodes total.}
\end{rem}

\section{Reflections}

The main purpose of this note is to provide yet another example of a simple research question that can be resolved with the aid of large language models. Although the result itself is likely of limited interest to experts, the example raises again a broader question: what are good introductory research problems for beginning researchers in an era in which LLMs can quickly dispatch many accessible problems? We leave this as the main open problem.

\bibliographystyle{alpha}
\bibliography{abb,ultimate}

\end{document}